\documentclass[12pt]{amsart}
\usepackage{amssymb,amsthm,amsmath,amsfonts}

\usepackage{geometry}
\usepackage{amssymb}
\usepackage{verbatim}
\usepackage{amsmath}
\usepackage[utf8]{inputenc}
\usepackage{enumerate}
\usepackage{verbatim}
\usepackage{amsfonts}
\usepackage{graphicx}
\usepackage{tikz}
\usetikzlibrary{matrix}
\usepackage{lscape}
\usepackage{color}

\newcommand{\R}{\mathbb{R}}

\newcommand{\mc}[1]{\mathcal{{#1}}}

\newcommand{\vp}{\varphi}

\newcommand{\dd}{\mathrm{d}}
\newcommand{\1}{\mathbf{1} }

\usepackage{color}

\DeclareMathOperator{\var}{Var}

\DeclareMathOperator{\supp}{supp}

\theoremstyle{definition}

\newtheorem{theorem}{Theorem}
\newtheorem{lemma}[theorem]{Lemma}
\newtheorem{proposition}[theorem]{Proposition}

\newtheorem{corollary}[theorem]{Corollary}

\theoremstyle{remark}
\newtheorem*{remark*}{Remark}

\theoremstyle{definition}

\begin{document}

\newgeometry{tmargin=2cm, bmargin=2cm, lmargin=2cm, rmargin=2cm}


\title[]{R\'enyi entropy and variance comparison for symmetric log-concave random variables}

\author{Maciej Białobrzeski}
\address{University of Warsaw}
\email{mb417513@students.mimuw.edu.pl}


\author{Piotr Nayar}
\thanks{P.N. was supported by the National Science Centre, Poland, grant 2018/31/D/ST1/01355}
\address{University of Warsaw}
\email{nayar@mimuw.edu.pl}

\begin{abstract}
   We show that for any $\alpha>0$ the R\'enyi entropy of order $\alpha$ is minimized, among all symmetric log-concave random variables with fixed variance, either for a uniform distribution or for a two sided exponential distribution. The first case occurs for $\alpha \in (0,\alpha^*]$ and the second case for $\alpha \in [\alpha^*,\infty)$, where $\alpha^*$ satisfies the equation $\frac{1}{\alpha^*-1}\log \alpha^*= \frac12 \log 6$, that is $\alpha^* \approx 1.241$. Using those results, we prove that one-sided exponential distribution minimizes R\'enyi entropy of order $\alpha \geq 2$ among all log-concave random variables with fixed variance.
\end{abstract}

\maketitle

{\footnotesize

}
\bigskip

\section{Introduction}\label{sec:intro}

For a random variable $X$ with density $f$ its R\'enyi entropy of order $\alpha \in (0,\infty) \setminus \{1\}$ is defined as 
\[
	h_\alpha(X)=h_\alpha(f) = \frac{1}{1-\alpha} \log\left( \int f^\alpha(x) \dd x \right),
\]  
assuming that the integral converges, see \cite{R61}. If $\alpha \to 1$ one recovers the usual Shannon differential entropy $h(f)=h_1(f)=-\int f \ln f$. Also, by taking limits one can define $h_0(f)=\log|\supp f|$, where $\supp f$ stand for the support of $f$ and $h_\infty(f)=-\log\|f\|_\infty$, there $\|f\|_\infty$ is the essential supremum of $f$.     

It is a well known fact that for any random variable  one has
\[
	h(X) \leq \frac12 \log \var(X) + \frac12 \log(2 \pi e) 
\]
with equality only for Gaussian random variables, see e.g. Theorem 8.6.5 in \cite{CT06}. The problem of maximizing R\'enyi entropy under fixed variance has been considered independently by Costa, Hero and Vignat in \cite{CHV03} and by Lutwak, Yang and Zhang in \cite{LYZ05}, where the authors showed, in particular, that for $\alpha \in (\frac13,\infty) \setminus \{1\}$ the maximizer is of the form 
\[
f(x) = c_0(1+(1-\alpha)(c_1 x)^2)_+^{\frac{1}{\alpha-1}},
\]   
which will be called the \emph{generalized Gaussian density}.
Any density satisfying $f(x) \sim x^{-3}(\log x)^{-2}$ shows that for $\alpha \leq \frac13$ the supremum of $h_\alpha$ under fixed variance is infinite. One may also ask for reverse bounds. However, the infimum of the functional $h_\alpha$ under fixed variance is $-\infty$ as can be seen by considering $f_n(x)=\frac{n}{2}\1_{[1,1+n^{-1}]}(|x|)$ for which the variance stays bounded  whereas $h_\alpha(f_n) \to -\infty$ as $n \to \infty$. Therefore, it is natural to restrict the problem to a certain natural class of densities, in which the R\'enyi entropy remains lower bounded in terms of the variance. In this context it is natural to consider the class of log-concave densities, namely densities having the form $f=e^{-V}$, where $V:\R \to (-\infty,\infty]$ is convex. In \cite{MNT21} it was proved that for any symmetric log-concave random variable one has
\[
	h(X) \geq \frac12 \log \var(X) + \frac12 \log 12
\]   
with equality  if and only if $X$ is a uniform random variable. In the present article we shall extend this result to general R\'enyi entropy. Namely, we shall prove the following theorem.

\begin{theorem}\label{thm:main}
Let $X$ be a symmetric log-concave random variable and $\alpha > 0$, $\alpha \neq 1$. Define $\alpha^*$ to be the unique solution to the equation $\frac{1}{\alpha-1}\log \alpha= \frac12 \log 6$ ($\alpha^* \approx 1.241$). Then
\[ h_{\alpha}(X)  \geq \frac12 \log \var(X) +  \frac12 \log 12 \text{ \qquad for } \alpha \leq \alpha^*	\]
and
\[ h_{\alpha}(X)  \geq \frac12 \log \var(X) + \frac12 \log2+\frac{\log\alpha}{\alpha-1}\text{ \qquad for } \alpha \geq \alpha^*.	\]
For $\alpha < \alpha^*$ equality holds if and only if $X$ is uniform random variable on a symmetric interval, while for $\alpha > \alpha^*$ the bound is attained only for two-sided exponential distribution. When $\alpha=\alpha^*$, two previously mentioned densities are the only cases of equality.
\end{theorem}

The above theorem for $\alpha <1$ trivially follows from the case $\alpha=1$ as already observed in \cite{MNT21} (see Theorem 5 therein). This is due to the monotonicity of R\'enyi entropy in $\alpha$. As we can see the case $\alpha \in [1,\alpha^*]$ of Theorem \ref{thm:main} is a strengthening of the main result of  \cite{MNT21}, as   in this case $h_\alpha(X) \leq h(X)$ and the right hand sides are the same. 

It turns out that Theorem \ref{thm:main} allows to deal with the non-symmetric case in the range $\alpha \geq 2$. The following corollary of our main theorem has been kindly communicated to us by Jiange Li.

\begin{corollary}\label{cor:non-symmetric}
Let $X$ be a log-concave random variable and let $\alpha \geq 2$. Then
\[ h_{\alpha}(X)  \geq \frac12 \log \var(X) +\frac{\log\alpha}{\alpha-1}	\]
with equality for one-sided exponential random variable. 
\end{corollary}

\noindent To prove it we shall use Theorem 6.1 from \cite{MT21}: for any iid log-concave random variables and $\alpha \geq 2$ one has $h_\alpha(X-Y)\leq h_\alpha(X)+\log 2$. Since $X-Y$ is log-concave and symmetric, we obtain
\[
	h_\alpha(X) \geq h_\alpha(X-Y)-\log 2 \geq \frac12 \log \var(X-Y) - \frac12 \log 2 + \frac{\log \alpha}{\alpha-1} = \frac12 \log \var(X) +\frac{\log\alpha}{\alpha-1}.
\]  
We remark that the problem of minimizing the R\'enyi entropy of order $\alpha \in (0,2)$ under fixed variance is open in the class of arbitrary log-concave densities (not necessarily symmetric). 

This article is organized as follows. In Section \ref{sec:reduction} we reduce Theorem \ref{thm:main} to the case $\alpha=\alpha^*$. In Section \ref{sec:degrees-of-freedom} we further simplify the problem by reducing it to \emph{simple} functions via the concept of degrees of freedom. Section \ref{sec:proof} contains the proof for these simple functions. In the last section we derive two applications of our main result.


\section{Reduction to the case $\alpha=\alpha^*$}\label{sec:reduction}

The following lemma is well known. We present its proof for completeness. The proof of point (ii) is taken from \cite{FMW16}. As pointed out by the authors, it can also be derived from Theorem 2 in \cite{B73} or from Theorem VII.2 in \cite{BM11}.  

\begin{lemma}\label{lem:log-concavity}
Suppose $f$ is a probability density in $\R^n$. 
\begin{itemize}
\item[(i)] The function $p \mapsto \int f^p$ is log-convex on $(0,\infty)$.
\item[(ii)] If $f$ is log-concave then the function $p \mapsto p^n \int f^p$ is log-concave on $(0,\infty)$.
\end{itemize}  
\end{lemma}

\begin{proof}
(i) This is a simple consequence of H\"older's inequality.

(ii) Let $\psi(p)=p^n \int f^p(x) \dd x$. The function $f$ can be written as $f=e^{-V}$, where $V:\R^n \to (-\infty,+\infty]$ is convex. Changing variables we get $\psi(p)=\int e^{-p V(\frac{z}{p}) }\dd z$. For any convex $V$ the so-called \emph{perspective function} $W(z,p)= pV(\frac{z}{p})$ is convex on $\R^n \times (0,\infty)$. Indeed, for $\lambda \in [0,1]$, $p_1,p_2 > 0$ and $z_1, z_2 \in \R^n$ we have
\begin{align*}
	W(\lambda z_1 & + (1-\lambda) z_2, \lambda p_1+(1-\lambda)p_2)   = (\lambda p_1+(1-\lambda)p_2) V\left( \frac{\lambda p_1 \frac{z_1}{p_1}+(1-\lambda)p_2 \frac{z_2}{p_2}}{\lambda p_1+(1-\lambda)p_2} \right)  \\
	& \leq \lambda p_1 V\left(\frac{z_1}{p_1} \right)+(1-\lambda) p_2  V\left(\frac{z_2}{p_2} \right) = \lambda W(z_1,p_1)+(1-\lambda) W(z_2,p_2).
\end{align*}
Since $\psi(p)=\int e^{-W(z,p)} \dd z$, the assertion follows from the Pr\'ekopa’s theorem from \cite{P73} saying that a marginal of a log-concave function is again log-concave.  
\end{proof}

\begin{remark*}
The use of the term \emph{perspective function} appeared in \cite{H-UL93}, however the convexity of this function was  known much earlier.    
\end{remark*}

The next corollary is a simple consequence of Lemma \ref{lem:log-concavity}. The right inequality of this corollary appeared in \cite{FMW16}, whereas the left inequality is classical.

\begin{corollary}\label{cor:monot-ent}
Let $f$ be a log-concave probability density in $\R^n$. Then for any $p \geq q > 0$ we have
\[
	0 \leq h_q(f)-h_p(f) \leq n \frac{\log q}{q-1} - n \frac{\log p}{p-1}.
\]
In fact, the first  inequality is valid without the log-concavity assumption.
\end{corollary}

\begin{proof}
To prove the first inequality we observe that due to Lemma \ref{lem:log-concavity} the function defined by  $\phi_1(p)=(1-p)h_p(f)$ is convex. From the monotonicity of slopes of $\phi_1$ we get that $\frac{\phi_1(p)-\phi_1(1)}{p-1} \geq \frac{\phi_1(q)-\phi_1(1)}{q-1}$, which together with the fact that $\phi_1(1)=0$ gives $h_p(f) \leq h_q(f)$.  

Similarly, to prove the right inequality we note that $\phi_2(p) = n \log p + (1-p)h_p(f)$ is concave with $\phi_2(1)=0$. Thus $\frac{\phi_2(p)-\phi_2(1)}{p-1} \leq \frac{\phi_2(q)-\phi_2(1)}{q-1}$  gives $\frac{n \log p}{p-1}-h_p(f) \leq \frac{n \log q}{q-1}-h_q(f)$, which finishes the proof.
\end{proof}

Having Corollary \ref{cor:monot-ent} we can easily reduce Theorem \ref{thm:main} to the case $\alpha=\alpha^*$. Indeed, the case $\alpha < \alpha^*$ follows from the left inequality of Corollary \ref{cor:monot-ent} ($h_p$ is non-increasing in $p$). The case $\alpha>\alpha^*$ is a consequence of the right inequality of the above corollary, according to which the quantity $h_\alpha(X)-\frac{\log \alpha}{\alpha-1}$ is non-decreasing in $\alpha$.

\section{Reduction to simple functions via degrees of freedom}\label{sec:degrees-of-freedom}

The content of this section is a rather straightforward adaptation of the method from \cite{MNT21}. Therefore, we shall only sketch the arguments. \\

By a standard approximation argument it is enough to prove our inequality for functions from the set $\mc{F}_L$ of all continuous even log-concave probability densities supported on $[-L,L]$. Thus, it suffices to show that
\begin{equation}\label{eq:inf}
	\inf \ \{ h_{\alpha^*}(f): \ f \in \mc{F}_L, \ \var(f)=\sigma^2  \} \geq \log \sigma + \frac12 \log 2 + \frac{\log \alpha^*}{\alpha^*-1}.
\end{equation}
Take $A=\{f \in \mc{F}_L: \var(f)=\sigma^2\}$. We shall show that $\inf_{f \in A} h_{\alpha^*}(f)$ is attained on $A$. Equivalently, since $\alpha^*>1$ it suffices to show that $M=\sup_{f \in A} \int f^{\alpha^*}$ is attained on $A$. We first argue that this supremum is finite. This follows from the estimate $\int f^{\alpha^*} \leq 2L f(0)^{\alpha^*}$ and from the inequality $f(0) \leq \frac{1}{\sqrt{2 \var(f)}} = \frac{1}{\sqrt{2} \sigma}$, see Lemma 1 in \cite{MNT21}. Next, let $(f_n)$ be a sequence of functions from $A$ such that $\int f_n^{\alpha^*} \to M$. According to Lemma 2 from \cite{MNT21}, by passing to a subsequence one can assume that $f_n \to f$ pointwise, where $f$ is some function from $A$. Since $f_n \leq f_n(0)\leq \frac{1}{\sqrt{2}\sigma}$, by the Lebesgue dominated convergence theorem we get that $\int f_n^{\alpha^*} \to  \int f^{\alpha^*}=M$ and therefore the supremum is attained on $A$.

Now, we say that $f \in A$ is an \emph{extremal point} in $A$ if $f$ cannot be written as a convex combination of two different functions from $A$, that is, if $f = \lambda f_1 + (1-\lambda) f_2$ for some $\lambda \in (0,1)$ and $f_1, f_2 \in A$, then necessarily $f_1=f_2$. It is easy to observe that if $f$ is not extremal, then it cannot be a maximizer of $\int f^{\alpha^*}$ on $A$. Indeed, if $f = \lambda f_1 + (1-\lambda) f_2$ for some $\lambda \in (0,1)$ and $f_1, f_2 \in A$ with $f_1 \ne f_2$, then the strict convexity of $x \to x^{\alpha^*}$ implies
\[
	\int f^{\alpha^*} = \int(\lambda f_1 + (1-\lambda) f_2)^{\alpha^*} < \lambda \int f_1^{\alpha^*} + (1-\lambda) \int f_2^{\alpha^*} \leq M.
\]      
This shows that in order to prove \eqref{eq:inf} it suffices to consider only the functions $f$ being extremal points of $A$. Finally, according to Steps III and IV of the proof of Theorem 1 from \cite{MNT21} these extremal points are of the form
\[
	 f(x)= c \1_{[0,a]}(|x|)+ce^{-\gamma(|x|-a)}\1_{[a,a+b]}(|x|), \qquad a+b=L, \ c>0, \ a,b,\gamma \geq 0,
\]
where it is also assumed that $\int f=1$. 

\section{Proof for the case $\alpha=\alpha^*$} \label{sec:proof}


Due to the previous section, we can restrict ourselves to probability densities $f$ of the form
\[ f(x)= c \1_{[0,a]}(|x|)+ce^{-\gamma(|x|-a)}\1_{[a,a+b]}(|x|), \qquad a,b,\gamma \geq 0. 
\]
The inequality is invariant under scaling  $f(x) \mapsto \lambda f(\lambda x)$
for any positive $\lambda$, so we can assume that $\gamma=1$ (note that in the case $\gamma=0$ we get equality). We have
\[ \int_\R f^\alpha = c^\alpha \int_\R  \1_{[0,a]}(|x|)+ c^\alpha \int_\R e^{-\alpha x}\1_{[0,b]}(|x|)=2c^\alpha \left(a+\frac{1-e^{-\alpha b}}{\alpha}\right) \]
and thus 
\[h_\alpha(f)=\frac{1}{1-\alpha}\log{\int_\R f^{\alpha}}=\frac{1}{1-\alpha}\log \left(2c^\alpha \left(a+\frac{1-e^{-\alpha b}}{\alpha}\right) \right). \]
Moreover,
\[ \var(f)= 2c \int_\R x^2\1_{[0,a]}(x)\dd x + 2c \int_\R (x+a)^2 e^{-x} \1_{[0,b]}\dd x 
 =2c\left (\frac{a^3}{3}+\int_0^b(x+a)^2e^{-x}\dd x \right), \]
so our inequality can be rewritten as 
\[ 
\frac{1}{1-\alpha^*}\log \left(2c^{\alpha^*} \left(a+\frac{1-e^{-\alpha^* b}}{\alpha^*}\right) \right) +\frac{\log \alpha^*}{1-\alpha^*} \geq \frac12 \log\left(2c\left (\frac{a^3}{3}+\int_0^b(x+a)^2e^{-x}\dd x \right) \right) +\frac12 \log 2, 
\]
which is 
\[
	\frac{1}{1-\alpha^*}\log \left(2c^{\alpha^*} \left(a \alpha^* +1-e^{-\alpha^* b}\right) \right)  \geq \frac12 \log\left(2c\left (\frac{a^3}{3}+\int_0^b(x+a)^2e^{-x}\dd x \right) \right) +\frac12 \log 2.
\]
The constraint $\int_\R f=1$ gives $c=\frac12 (a+1-e^{-b})^{-1}$. After multiplying both sides by $2$, exponentiating both sides and plugging the expression for $c$ in, we get the equivalent form of the inequality, $G(a,b,\alpha^*) \geq 0$, where
\begin{equation}\label{fundamental}
 G(a,b,\alpha)= 2 (a\alpha +1-e^{-\alpha b})^{\frac{2}{1-\alpha}}(a+1-e^{-b})^{\frac{1-3\alpha}{1-\alpha}} - \left(\frac{a^3}{3}+\int_0^b(x+a)^2e^{-x}\dd x \right). 
 \end{equation}
We will also write $G(a,b)=G(a,b,\alpha^*)$. 

To finish the proof we shall need the following lemma.

\begin{lemma}\label{lem:tech}
The following holds:
\begin{itemize}
\item[(a)] $\frac{\partial^4 }{\partial a^4} G(a,b) \geq 0$ holds for every $a,b \geq 0$, 
\item[(b)] $\lim_{a \to \infty} \frac{\partial^3 }{\partial a^3} G(a,b) = 0$ for every $b \geq 0$,
\item[(c)] $\lim_{a \to \infty} \frac{\partial^2 }{\partial a^2} G(a,b) \geq 0$ for every $b \geq 0$,
\item[(d)] $\frac{\partial }{\partial a} G(a,b) \big|_{a=0} \geq 0$ for every $b \geq 0$,
\item[(e)] $G(0,b) \geq 0$ for every $b \geq 0$.
\end{itemize}
\end{lemma}

\noindent With these claims at hand it is easy to conclude the proof. Indeed, one easily gets, one by one, 
\[
	\frac{\partial^3 }{\partial a^3} G(a,b) \leq 0, \qquad \frac{\partial^2 }{\partial a^2} G(a,b) \geq 0, \qquad \frac{\partial }{\partial a} G(a,b) \geq 0, \qquad G(a,b) \geq 0, \qquad b \geq 0. 
\]

The proof of points (d) and (e) relies on the following simple lemma.

\begin{lemma}\label{lem:sign}
Let $f(x)=\sum_{n=0}^{\infty}a_n x^n$, where the series is convergent for every nonnegative $x$. If there exists a nonnegative integer $N$ such that $a_n \geq 0$ for $n<N$ and $a_n \leq 0$ for $n \geq N$, then $f$ changes sign on $(0,\infty)$ at most once. Moreover, if at least one coefficient $a_n$ is positive and at least one negative, then there exists $x_0$ such that $f(x) > 0$ on $[0,x_0)$ and $f(x) < 0$ on $(x_0,\infty)$.  
\end{lemma}
\begin{proof}
Clearly the function $f(x)x^{-N}$ is nonincreasing on $(0,\infty)$, so the first claim follows. To prove the second part we observe that for small $x$ the function $f$ must be strictly positive and $f(x)x^{-N}$ is strictly decreasing on $(0,\infty)$.   
\end{proof}

With this preparation we are ready to prove Lemma \ref{lem:tech}.

\begin{proof}[Proof of Lemma \ref{lem:tech}] \

(a) This point is the crucial observation of the proof. It turns out that 
\begin{align*}
 \frac{\partial^4 G}{\partial a^4}(a,b,\alpha) &= 8\alpha(\alpha+1)(3\alpha-1)(1+a-e^{-b})^{\frac{3\alpha-1}{\alpha-1}}(1+a\alpha-e^{-b\alpha})^{\frac{2}{1-\alpha}}   \\
& \qquad \qquad \times \left(\frac{(e^b-\alpha e^{b\alpha}+(\alpha-1)e^{b+b\alpha})}{(\alpha-1)(e^b(a+1)-1)(e^{b\alpha}(a\alpha+1)-1)}\right)^4,
\end{align*}
which is nonegative for $\alpha> \frac13$. \\

(b) By a direct computation we have

\begin{align*}
\frac{\partial^3 G(a,b,\alpha)}{\partial a^3} &= -2 - \frac{4\alpha}{(1-\alpha)^3}(1+a-e^{-b})^{\frac{2}{\alpha-1}}(1+a\alpha-e^{-b \alpha})^{\frac{1-3\alpha}{\alpha-1}} \\ & \quad \times [ (\alpha+1)(3\alpha-1)(1+a\alpha-e^{-b\alpha})^3-2\alpha^3(\alpha+1)(1+a-e^{-b})^3 \\& \qquad \qquad +3\alpha(\alpha+1)(3\alpha-1)(1+a-e^{-b})^2(1+a\alpha-e^{-b\alpha}) \\ & \qquad \qquad +6\alpha(1-3\alpha)(1+a-e^{-b})(1+a\alpha-e^{-b\alpha})^2 ].
\end{align*}
When $a$ tends to infinity with $b$ fixed this converges to
\[
	-2 - \frac{4\alpha}{(1-\alpha)^3} \alpha^{\frac{1-3\alpha}{\alpha-1}}\left( (\alpha+1)(3\alpha-1)\alpha^3 - 2 \alpha^3 (\alpha+1)+3\alpha^2 (\alpha+1)(3\alpha-1)  + 6 \alpha^3(1-3\alpha) \right),
\]  
which is
$
	-2 + 12 \alpha^3 \alpha^{\frac{1-3\alpha}{\alpha-1}}
$. If $\alpha=\alpha^*$, using equality $(\alpha^*)^{\frac{2}{1-\alpha^*}}=\frac16$, we get that this expression is equal to $0$.  \\

(c) Again a direct computation yields

\begin{align*}
\frac{\partial^2 G(a,b,\alpha)}{\partial a^2} &=
\frac{4\alpha^2(\alpha+1)(a-e^{-b}+1)\left(1-e^{-b}a^{-1}+a^{-1}\right)^{\frac{2\alpha}{\alpha-1}}\left(\alpha-e^{-\alpha b}a^{-1}+a^{-1}\right)^{-\frac{2\alpha}{\alpha-1}}}{(\alpha-1)^2} \\
 & \qquad +\frac{4\alpha(3\alpha-1)(a-e^{-b}+1)\left(1-e^{-b}a^{-1}+a^{-1}\right)^{\frac{2}{\alpha-1}}\left(\alpha-e^{-\alpha b}a^{-1}+a^{-1}\right)^{-\frac{2}{\alpha-1}}}{(\alpha-1)^2} \\
 & \qquad +\frac{8\alpha(1-3\alpha)(a-e^{-b}+1)\left(1-e^{-b}a^{-1}+a^{-1}\right)^{\frac{\alpha+1}{\alpha-1}}\left(\alpha-e^{-\alpha b}a^{-1}+a^{-1}\right)^{-\frac{\alpha+1}{\alpha-1}}}{(\alpha-1)^2} \\ & \qquad  -2a+2e^{-b}-2.
\end{align*}
As $a$ tends to infinity, we have 
\[
(1-e^{-b}a^{-1}+a^{-1})^w=1+w(1-e^{-b}) a^{-1}+o(a^{-1})
\]
and
\[
 (\alpha-e^{-\alpha b}a^{-1}+a^{-1})^w=\alpha^{w}+w(1-e^{-\alpha b})\alpha^{w-1} a^{-1}+o(a^{-1}).
\]
Using these formulas together with the above expression for the second derivative easily gives 
\[
	 \frac{\partial^2 G(a,b,\alpha)}{\partial a^2} = h_1(\alpha) \frac{1}{x} + h_2(b,\alpha) + o(a^{-1}),
\]
where
\[
	h_1(\alpha) = 12 \alpha^{-\frac{2}{\alpha-1}}-2
\]
and 
\[
	h_2(b, \alpha) =2(e^{-b}-1) + \frac{4 \alpha \left(\alpha^{\frac{1}{1-\alpha}}-\alpha^{\frac{\alpha}{1-\alpha}}\right)^2}{(\alpha-1)^3}\left( 2   \left(\alpha-1 -\alpha e^{-b}+e^{-b \alpha}\right)+3  \left(1-e^{-b}\right) \alpha (\alpha-1) \right).
\]
We have $h_1(\alpha^*)=0$. Moreover,
\[
	\frac{4 \alpha^* \left((\alpha^*)^{\frac{1}{1-\alpha^*}}-(\alpha^*)^{\frac{\alpha^*}{1-\alpha^*}}\right)^2}{(\alpha^*-1)^3} = \frac{4\alpha^* \left( \frac{1}{\sqrt{6}}- \frac{1}{\sqrt{6} \alpha^*} \right)^2}{(\alpha^*-1)^3} = \frac{2}{3\alpha^*(\alpha^*-1)}.
\]
Hence,
\[
	\lim_{a \to \infty}  \frac{\partial^2 G(a,b,\alpha)}{\partial a^2} = h_2(b,\alpha^*) = \frac{4}{3\alpha^*(\alpha^*-1)} \left( (1-e^{-b}) \alpha^* - (1-e^{-b \alpha^*})  \right) .
\]
This expression is nonnegative for $b \geq 0$ since the function $h_3(x) = 1-e^{-x}$ is  concave, so we have $\frac{1-e^{-b}}{b} =  \frac{h_3(b)}{b} \geq  \frac{h_3(\alpha^* b)}{\alpha^* b} =  \frac{1-e^{-\alpha^* b}}{\alpha^* b}$ as $\alpha^*>1$ (monotonicity of slopes). \\

(e) To illustrate our method, before proceeding with the proof of (d) we shall prove (e), as the idea of the proof of (d) is similar, but the details are more complicated. Our goal is to show the inequality
\begin{equation} (1-e^{-\alpha^* b})^{\frac{2}{1-\alpha^*}}(1-e^{-b})^{\frac{1-3\alpha^*}{1-\alpha^*}} \geq 1-\frac{b^2+2b+2}{2}e^{-b}. \label{zero_b} \end{equation}
after taking the logarithm of both sides our inequality reduces to nonnegativity of
\[ \phi(b) = \frac{2}{1-\alpha^*} \log(1-e^{-\alpha^* b}) +\frac{1-3\alpha^*}{1-\alpha^*} \log(1-e^{-b}) - \log\left(1-\frac{b^2+2b+2}{2}e^{-b}\right). \]
We have
\begin{equation*} \phi'(b)=\frac{2\alpha^*}{(1-\alpha^*)(e^{\alpha^* b}-1)} +\frac{1-3\alpha^*}{(1-\alpha^*)(e^{b}-1)}+\frac{b^2}{b^2+2b-2e^{b}+2}. \end{equation*}
It turns out that $\phi(b)$ changes sign on $(0,\infty)$ at most once. To show that, firstly, clear out the denominators (they have fixed sign on $(0,\infty)$) to obtain the expression 
\begin{equation} 2\alpha^*(b^2+2b-2e^{b}+2)(e^b-1)+ (1-3\alpha^*)(e^{\alpha^* b}-1)(b^2+2b-2e^b+2)+b^2(1-\alpha^*)(e^b-1)(e^{\alpha^* b}-1). \label{logderivative} \end{equation}

Now we will apply Lemma $\ref{lem:sign}$ to $\eqref{logderivative}$. That expression can be rewritten as

\[-4 \alpha^* \left(\sum_{n=3}^{\infty}\frac{b^n}{n!}\right)\left(\sum_{n=1}^{\infty} \frac{b^n}{n!}\right) +\left(6\alpha^*-2\right)\left(\sum_{n=1}^{\infty} \frac{(\alpha^*b)^n}{n!}\right)\left(\sum_{n=3}^{\infty} \frac{b^n}{n!}\right)+b^2(1-\alpha^*)\left(\sum_{n=1}^\infty \frac{b^n}{n!} \right)\left(\sum_{n=1}^\infty \frac{(\alpha^* b)^n}{n!}\right), \]
so the $n$-th coefficient $a_{n}$ in the Taylor expansion is equal to
\begin{align*}
a_{n} &= (6\alpha^*-2)\left(\sum_{j=1}^{n-3} \frac{(\alpha^*)^j}{j!(n-j)!}\right) - 4\alpha^* \left(\sum_{j=1}^{n-3} \frac{1}{j!(n-j)!}\right) +(1-\alpha^*)\left(\sum_{j=1}^{n-3} \frac{(\alpha^*)^j}{j!(n-2-j)!}\right) \\
& \leq  \frac{1}{n!} (6\alpha^*-2)(\alpha^*+1)^n +\frac{1-\alpha^*}{(n-2)!}\left( (\alpha^*+1)^{n-2} - 1-(\alpha^*)^{n-2} \right)  \\
& \leq \frac{6}{n!}(\alpha^*+1)^n-\frac{n(n-1)}{30n!}(\alpha^*+1)^n+\frac{8n^2}{n!}(\alpha^*)^{n} .
\end{align*}
When $n \geq 17$, we have $\frac{n(n-1)}{30} > 7$ and $(\frac{\alpha^*+1}{\alpha^*})^n \geq (\frac85)^n \geq 8 n^2$, so $a_n$ is less than zero for $n \geq 17$. It can be checked (preferably using computational software) that the rest of coefficients $a_n$ satisfy the pattern from Lemma $\ref{lem:sign}$, with $a_n=0$ for $n \leq 4$, $a_n>0$ for $n=5,6,7$ and $a_n<0$ for $n \geq 8$. 

This way we have proved that $\phi'(b)$ changes sign in exactly one point $x_0 \in (0, \infty)$. Thus, $\phi$ is first increasing and then decreasing. Since $\phi(0)=0$ and $\lim_{b \to \infty} \phi(b)=0$, the assertion follows. \\

(d) We have to show that
\[ \frac{(1-e^{-b})^{\frac{2\alpha^*}{\alpha^*-1}}(1-e^{-b\alpha^*})^{-\frac{1+\alpha^*}{\alpha^*-1}}}{\alpha^*-1}[(3\alpha^*-1)(1-e^{-b\alpha^*})-2\alpha^*(1-e^{-b})] \geq 1-(b+1)e^{-b}. \]
Let $\vp_1(b)$ be the expression on the left side and $\vp_2(b)$ on the right side. Both $\vp_1$ and $\vp_2$ are positive for $b > 0$, so we can take the logarithm of both sides.  We will now show that $(\log(\vp_1))'-(\log(\vp_2))'$ changes sign at most once on $(0,\infty)$. We have

\begin{align*} 
\label{jakasnazwa} (\log(\vp_1))'-(\log(\vp_2))'&=\frac{2\alpha^*}{\left(e^b-1\right)(\alpha^*-1)}-\frac{(\alpha^*+1)\alpha^*}{(\alpha^*-1)\left(e^{b\alpha^*}-1\right)}
\\& \qquad + \frac{\alpha^*(3\alpha^*-1)e^{b}-2e^{b\alpha^*}\alpha^*}{e^b(1-3\alpha^*)+2\alpha^* e^{b\alpha^*}+(\alpha^*-1)e^{b\alpha^*+b}}-\frac{b}{e^b-b-1}.
\end{align*}
Multiplying the  above expression by the product of denominators does not change the hypothesis, since each of the denominators  is positive. After this multiplication we get the expression

\begin{align*}
& [-\left(e^b-1\right)\left(e^b-1-b \right)(\alpha^*+	1)\alpha^*+2\left(e^b-1-b\right)\alpha^*\left(e^{b\alpha^*}-1\right)-b\left(e^b-1\right)(\alpha^*-1)\left(e^{b\alpha^*}-1\right)] \\ &
\qquad \qquad  \times \left(e^b(1-3\alpha^*)+2\alpha^* e^{b\alpha^*}+(\alpha^*-1)e^{b\alpha^*+b}\right)
 \\
  & \qquad + \alpha^* (\alpha^*-1)  \left(e^b-1\right)\left(e^b-1-b\right)\left(e^{b\alpha^*}-1\right)\left(e^b(3\alpha^*-1)-2e^{b\alpha^*}\right).
\end{align*}
Let us consider the Taylor series $\sum_{n \geq 0} a_n b^n$ of this function (it is clear that the series converges to the function everywhere). It can be shown (again  using computational software) that coefficients of this series up to order $9$ are nonnegative and coefficients of order greater than $9$, but lesser than $30$ are negative. Now we will show negativity of coefficients of order at least $30$ (our bound will be very crude, so it would not work, if we replaced $30$ with lower number). Firstly we note that 
\[
e^b(1-3\alpha^*)+2\alpha^* e^{b\alpha^*}+(\alpha^*-1)e^{b\alpha^*+b}
\]
has $n$-th Taylor coefficient equal to 
\[
\frac{1-3\alpha^* + 2(\alpha^*)^{n+1}+(\alpha^*-1)(\alpha^*+1)^n}{n!} \geq \frac{1-3\alpha^* + 2\alpha^*+\alpha^*-1}{n!} = 0,
\]
 so all its coefficients are nonnegative. Thus we can change expression in square brackets to 
$(e^b-1)(e^{b \alpha^*-1})(5/2- b/5)$ (we discard the first term and bound from above the second and third one) to increase every Taylor coefficient of main expression. Now we want to show the negativity of coefficients of order at least $30$ for 
\[  (e^b-1)(e^{b\alpha^*}-1)[(5/2-b/5)(e^b(1-3\alpha^*)+2\alpha^* e^{b\alpha^*}+(\alpha^*-1)e^{b\alpha^*+b})+\alpha^*(\alpha^*-1)(e^b-b-1)((3\alpha^*-1)e^b-2e^{b\alpha^*})] \] 
The expression in square brackets has $n$-th Taylor coefficient $c_n$ equal to zero for $n \in \{0,1\}$, while for $n \geq 2$ it is
\begin{align*}
 c_n &= \frac{5(1-3\alpha^*)}{2n!}+\frac{3\alpha^*-1}{5(n-1)!}+\frac{5(\alpha^*)^{n+1}}{n!}-\frac{2(\alpha^*)^n}{5(n-1)!}+\frac{5(\alpha^*-1)(\alpha^*+1)^n}{2n!}-\frac{(\alpha^*-1)(\alpha^*+1)^{n-1}}{5(n-1)!}  \\& \qquad + \alpha^*(\alpha^*-1)(3\alpha^*-1) \frac{2^n-n-1}{n!} -\frac{2\alpha^*(\alpha^*-1)}{n!}((\alpha^*+1)^n-(\alpha^*)^n-n(\alpha^*)^{n-1}) .
\end{align*}
Using the bounds
\[
	\frac{5(1-3\alpha^*)}{2n!} \leq 0, \qquad -\frac{2(\alpha^*)^n}{5(n-1)!} \leq 0, \qquad  \alpha^*(\alpha^*-1)(3\alpha^*-1) \frac{2^n-n-1}{n!} \leq \frac{2^n}{n!}
\]
and
\[
	 \frac{2\alpha^*(\alpha^*-1)}{n!}((\alpha^*)^n+n (\alpha^*)^{n-1}) \leq \frac{(n+1)(\alpha^*)^n}{n!}, \qquad 
	 -\frac{(\alpha^*-1)(\alpha^*+1)^{n-1}}{5(n-1)!} \leq -\frac{\frac45 n}{10 n!}(\alpha^*-1)(\alpha^*+1)^n
\]
we get the following  upper bound for $c_n$ for $n \geq 2$
\begin{align*}
c_n &\leq \frac{(3\alpha^*-1)}{5(n-1)!}+\frac{5(\alpha^*)^{n+1}}{n!}+\frac{2^n}{n!}+\frac{(n+1)(\alpha^*)^n}{n!} +\frac{(\alpha^*+1)^n(\alpha^*-1)(25-20\alpha^*-4n/5)}{10n!}  \\
& \leq \frac{(n+8)(\alpha^*)^n}{n!}+\frac{n+2^{n}}{n!} +\frac{(\alpha^*+1)^n(1-3n)}{200n!},
\end{align*}
since $\frac{1}{10}(\alpha^*-1)(25-20 \alpha^*) \leq \frac{1}{200}$ and $\frac{4}{50}(\alpha^*-1) \geq \frac{3}{200}$.
This bound works for $n \in \{0,1\}$, too. We have
\[ (e^b-1)(e^{b\alpha^*}-1)=\sum_{n=2}^{\infty}b^n \frac{(\alpha^*+1)^n-(\alpha^*)^n-1}{n!}, \] so $(e^b-1)(e^{b\alpha^*}-1)$ has nonnegative coefficients.
Now we can  bound the Taylor series coefficients $d_n$ of the main expression as follows
\begin{equation*}
\begin{split} d_n &\leq \frac{1}{n!} \sum_{k=0}^{n-2} \binom{n}{k} \left( (k+8)(\alpha^*)^k + k+2^k +\left(\alpha^*+1\right)^k\frac{1-3k}{200}\right)\left(\left(\alpha^*+1\right)^{n-k}-(\alpha^*)^{n-k}-1\right)
\end{split}
\end{equation*}
Changing the upper limit of the sum from $n-2$ to $n$ increases the sum for $n\geq 30$ -- for $k=n-1$ we have $(\alpha^*+1)^{n-k}-(\alpha^*)^{n-k}-1=0$ and the term for $k=n$ is surely positive for $ n \geq 30$, thus we have 
\begin{align*}
n! d_n \leq \sum_{k=0}^{n} &\binom{n}{k} \left( (k+8)(\alpha^*)^k + k+2^k +\left(\alpha^*+1\right)^k\frac{1-3k}{200}\right)\left(\left(\alpha^*+1\right)^{n-k}-(\alpha^*)^{n-k}-1\right) \leq  \\ &\leq (n+8)(2\alpha^*+1)^n+n(\alpha^*+2)^n+(\alpha^*+3)^n +\frac{1}{200}(2\alpha^*+2)^n \\ & \qquad - \frac{3n}{400}(2\alpha^*+2)^{n} + \frac{3n}{200}(2\alpha^*+1)^{n} + \frac{3n}{200}(\alpha^*+2)^{n},
\end{align*}
where we neglected all the negative terms except for the term $\sum_{k=0}^n {n \choose k}\frac{-3k}{200} (\alpha^*+1)^n =  - \frac{3n}{400}(2\alpha^*+2)^{n} $ and bounded $k$ by $n$ in all the positive terms (whenever $k$ appeared linearly). 

It is clear that negative term $-\frac{3n}{400}(2\alpha^*+2)^{n}$ dominates, so $d_n$ is negative when $n$ is sufficiently large. In fact, the expression is negative for $n \geq 30$. It is not hard to prove (again by checking some concrete values numerically and using convexity arguments) that for $n \geq 30$ we have
\[
	n+8+\frac{3}{200} n  < 0.104 \left( \frac{\alpha^*+3}{2\alpha^*+1} \right)^n, \qquad \left(1+\frac{3}{200} \right) n < 0.01 \left( \frac{\alpha^*+3}{2\alpha^*+1} \right)^n,
\]
so for $n \geq 30$ we have
\[
	n! d_n < 1.114 (\alpha^*+3)^n - \frac{3n-2}{400} (2\alpha^*+2)^n = (2\alpha^*+2)^n\left( 1.114 \left( \frac{\alpha^*+3}{2\alpha^*+2} \right)^n - \frac{3n-2}{400} \right)<0 .   
\]

From Lemma \ref{lem:sign} we get that $(\log(\vp_1))'-(\log(\vp_2))'$ on $(0,\infty)$ is first positive and then negative. This means that $(\log(\vp_1))-(\log(\vp_2))$ first increasing and then decreasing. In order to prove that it is everywhere nonnegative it suffices to check that it is nonnegative when $b \to 0^+$ and $b \to \infty$. The limit when $b \to \infty$ is easily seen to be $0$. To check the limit when $b \to 0^+$ it is enough check the Taylor expansion of $\phi_1(b)-\phi_2(b))$. Note that
\begin{align*}
	\frac{\phi_1(b)-\phi_2(b)}{b^2} & = \left(1- \frac{b}{2} + O(b^2) \right)^{\frac{2\alpha^*}{\alpha^*-1}} (\alpha^*)^{-\frac{1+\alpha^*}{\alpha^*-1}}\left( 1- \frac12 b \alpha^* + O(b^2) \right)^{-\frac{1+\alpha^*}{\alpha^*-1}} \\
	& \qquad \times \left(3\alpha^* - \frac{\alpha^*(2+3\alpha^*)}{2} b + O(b^2) \right) - \frac12 + \frac{b}{3} + O(b^2).
\end{align*}
By using the equality $(\alpha^*)^{\frac{2}{1-\alpha^*}} = \frac16$ we see that the constant term vanishes. In fact
\[
\frac{\phi_1(b)-\phi_2(b)}{b^2} = \left( \frac13 - (\alpha^*)^{\frac{2}{1-\alpha^*}} \right)b + O(b^2) = \frac16 b + O(b^2).
\]

\end{proof}

\section{Applications}

\subsection{Relative $\alpha$-entropy}\label{sec:q-entropy}

Recall that if $f_X$ denotes the density of a random variable $X$ then the relative $\alpha$-entropy studied by Ashok Kumar and Sundaresan in \cite{AS15} is defined as
\[
	I_{\alpha}(X \| Y) = \frac{\alpha}{1-\alpha} \log\left( \int \frac{f_X}{\| f_X\|_\alpha} \left( \frac{f_Y}{\|f_Y\|_\alpha} \right)^{\alpha-1} \right) 
\]
for $\alpha \in (0,1) \cup (1,\infty)$, where $\|f\|_\alpha= (\int |f|^\alpha)^{1/ \alpha}$. We shall derive an analogue of Corollary 5 from \cite{MNT21}. To this end we shall need the following fact.  

\begin{proposition}[\cite{AS15}, Corollary 13]\label{prop:1}
Suppose $\alpha>0$, $\alpha \neq 1$ and let $\mc{P}$ be the family of
probability measures such that the mean of the function $T : \R \to \R$ under them is fixed
at a particular value $t$. Let the random variable $X$ have a distribution from $\mc{P}$, and let
$Z$ be a random variable that maximizes the R\'enyi entropy of order $\alpha$ over $\mc{P}$. Then
\[
	I_\alpha(X \| Z) \leq h_\alpha(Z) - h_\alpha(X).
\]
\end{proposition}
Combining Proposition \ref{prop:1} with Theorem \ref{thm:main} and using expressions for the R\'enyi entropy and variance of a generalized Gaussian density derived in \cite{LYZ05}, one gets the following corollary.

\begin{corollary}\label{cor:relatice-q-entropy}
Suppose $\alpha>1$. Let $X$ be a symmetric log-concave real random variable. Let $Z$ be the random variable having generalized Gaussian density with parameter $\alpha$ and satisfying $\var(X)=\var(Z)$. Then $I_\alpha(X \| Z) \leq C(\alpha)$, where
\[
	 C(\alpha) = \log\left( (2\alpha)^{\frac{1}{1-\alpha}} (3\alpha-1)^{- \frac{1}{1-\alpha}}(\alpha-1)^{-\frac12} B\left(\frac12, \frac{\alpha}{\alpha-1} \right) \right)- \min\left( \frac12 \log12 , \frac12 \log 2 + \frac{\log \alpha}{\alpha-1} \right).
\] 
Here $B(a,b)=\frac{\Gamma(a)\Gamma(b)}{\Gamma(a+b)}$ stand for the Beta function.
\end{corollary}

\subsection{Reverse entropy power inequality}\label{sec:reverse-epi}

The R\'enyi entropy power of order $\alpha>0$ of a  random vector $X$ in $\R^n$ is defined as $N_\alpha(X)=\exp(\frac2n h_\alpha(X))$. We also write $N(X)$ for $N_1(X)$. If we combine our Theorem \ref{thm:main} with Theorem 2 from \cite{LYZ05}, we get the following sandwich bound for $\alpha>1$ and a symmetric log-concave random variable $X$,
\begin{equation} \label{ineq:ent-power}
	C_-(\alpha) \var(X) \leq N_\alpha(X) \leq C_+(\alpha) \var(X),
\end{equation}
where
\[
	C_-(\alpha) = \left\{\begin{array}{ll}
	12 & \alpha \in (1,\alpha^*) \\
	2 \alpha^{\frac{2}{\alpha-1}} & \alpha \geq \alpha^*
	\end{array}	\right.,
	   \qquad C_+(\alpha) = \frac{3\alpha-1}{\alpha-1} \left( \frac{2\alpha}{3\alpha-1} \right)^{\frac{2}{1-\alpha}}   B\left(\frac12, \frac{\alpha}{\alpha-1}\right)^2 .
\] 
Note that the case of $\alpha \in (\frac13,1]$ was discussed in \cite{MNT21}. We point out that for the upper bound the log-concavity assumption is not needed. Nevertheless, note that for $\alpha>1$ the so called generalized Gaussian density for which the right inequality is saturated, is symmetric and log-concave. 

We can now easily derive an analogue of Corollary 6 from \cite{MNT21} for $\alpha>1$.   

\begin{corollary}\label{cor:reverse-epi}
Let $\alpha \geq 1$. For $X,Y$ uncorrelated, symmetric real log-concave random variables one has
\[
	N_\alpha(X+Y) \leq \frac{C_+(\alpha)}{C_-(\alpha)} \left( N_\alpha(X)+ N_\alpha(Y) \right).
\]
\end{corollary}

\begin{proof}
We have
\[
	N_\alpha(X+Y) \leq C_+(\alpha) \var(X+Y) = C_+(\alpha)(\var(X)+\var(Y)) \leq \frac{C_+(\alpha)}{C_-(\alpha)} \left( N_\alpha(X)+ N_\alpha(Y) \right).
\]
\end{proof}

Using bounds from Corollary \ref{cor:non-symmetric}, inequalities analogous to \eqref{ineq:ent-power} and one from Corollary \ref{cor:reverse-epi} can be stated for $\alpha\geq2$ and an arbitrary log-concave random variable,
\[
	\frac12 C_-(\alpha) \var(X) \leq N_\alpha(X) \leq C_+(\alpha) \var(X).
\] 
\begin{corollary}\label{cor:reverse-epi-nonsym}
Let $\alpha \geq 2$. For $X,Y$ uncorrelated real log-concave random variables one has
\[
	N_\alpha(X+Y) \leq \frac{2C_+(\alpha)}{C_-(\alpha)} \left( N_\alpha(X)+ N_\alpha(Y) \right).
\]
\end{corollary}
Let us point out that many other reversals of the celebrated entropy power inequality (EPI) of Shannon and Stam \cite{S48, S59} has been established. Firstly,  one should point out that according to the work of Bobkov and Chistyakov \cite{BC15} no reverse EPI can be formulated for general independent random variables. Indeed, there exists $X$ with finite entropy and such that $h(X+Y)=\infty$ for every independent $Y$ with finite entropy. Bobkov and Madiman in \cite{BM11_2} showed that  for any pair $X,Y$ of independent log-concave random vectors in $\R^n$ there exist affine entropy preserving transformations $u,v:\R^n \to \R^n$ such that 
\[
	N(u(X)+v(Y)) \leq C(N(X)+N(Y)),
\]
where $C$ is a universal constant. This is sometimes called the \emph{positional} reverse entropy power inequality. An analogue of this result for R\'enyi entropy is given in \cite{MMX17}. 

In \cite{BNT16} Ball, Tkocz and the second named author showed that for any symmetric log-concave random vector $(X,Y)$ in $\R^2$ (in particular, for $X,Y$ being independent symmetric real random variables) one has $N^\frac12(X+Y) \leq e (N^\frac12(X)+N^\frac12(Y))$ and conjectured that the inequality holds with constant $1$ instead of $e$. The authors proved also that $N^\kappa(X+Y) \leq N^\kappa(X)+N^\kappa(Y)$ holds true with $\kappa=\frac{1}{10}$ in the above setting. In \cite{MMX17} Madiman, Melbourne and Xu established the same inequality for arbitrary R\'enyi entropy power of order $\alpha \geq 1$. In fact their constant $\kappa$ depends on $\alpha$ and is always better than $\kappa=\frac{1}{10}$. In \cite{L18} Li showed that $N_\alpha^{1/2}(X+Y) \leq N_\alpha^{1/2}(X)+ N_\alpha^{1/2}(Y)$ holds true with $p=0,2$. Marsiglietti and Kostina established the inequality $N(X+Y) \leq \frac{\pi e}{2}(N(X)+N(Y))$ for uncorrelated log-concave random variables. The constant was improved to $\frac{\pi e}{6}$ in the case of symmetric uncorrelated log-concave random variables by Madiman, Tkocz and the second named author. For reverse EPI for two iid summands see \cite{BM11, LMM20, MMX17}.

For results concerning the forward R\'enyi EPI see \cite{BC12, BC15, RS16, BM17, L18, MM19, XMM} . More information on various forward and reverse forms of the entropy power inequality can be found in the survey article \cite{MMX17}.

\subsection*{Acknowledgments} We would like to thank Jiange Li for communicated to us the fact that Theorem \ref{thm:main} combined with Theorem 6.1 from \cite{MT21} yields Corollary \ref{cor:non-symmetric}.

\end{document}